\pdfoutput=1
\documentclass[a4paper,UKenglish,cleveref,autoref,thm-restate]{lipics-v2021}
\usepackage{float}
\usepackage{xurl}
\usepackage[colorinlistoftodos,disable]{todonotes}
\usepackage[commandnameprefix=ifneeded,authormarkup=superscript]{changes}
\definechangesauthor[name=duc, color=blue]{duc}
\usepackage{xcolor}
\usepackage{xspace}
\usepackage{tcolorbox}
\tcbuselibrary{breakable}
\usetikzlibrary{arrows.meta,positioning,calc}
\setuptodonotes{inline}

\title{Concave is the New Linear: Sybil Attacks Collapse
Anti-Plutocratic DAO Voting}
\titlerunning{Concave is the New Linear}

\author{Austin Bennett}{Circle Research, USA}{gaustinbennett@gmail.com}{https://orcid.org/0009-0005-3186-9774}{} %
\author{Preston Vander Vos}{Circle Research, USA}{preston.vandervos@circle.com}{https://orcid.org/0009-0004-7206-5839}{}
\author{Mira Belenkiy}{Circle Research, USA}{mira.belenkiy@circle.com}{https://orcid.org/0009-0000-6801-7368}{}
\author{Duc V. Le}{Circle Research, USA}{levduc112@gmail.com}{https://orcid.org/0000-0002-8123-2713}{} %

\authorrunning{A. Bennett, P. Vander Vos, M. Belenkiy, and D.\,V. Le}
\Copyright{Austin Bennett, Preston Vander Vos, Mira Belenkiy, and Duc V. Le}


\category{}
\relatedversion{To appear in Financial Cryptography and Data Security (FC) 2027.}

\hideLIPIcs  %
\nolinenumbers %

\EventEditors{John Q. Open and Joan R. Access}
\EventNoEds{2}
\EventLongTitle{6th Conference on Advances in Financial Technologies (AFT 2026)}
\EventShortTitle{AFT 2026}
\EventAcronym{AFT}
\EventYear{2026}
\EventDate{October 2026}
\EventLocation{TBD}
\EventLogo{}
\SeriesVolume{42}
\ArticleNo{0}
\def\eg{\textit{e.g.}\xspace}
\def\etal{\textit{et~al.}\xspace}
\def\ie{\textit{i.e.}\xspace}

\newcommand{\pparagraph}[1]{\vspace{0.5em}\noindent\textbf{#1.}}

\newtcolorbox{insightbox}[1][Takeaway]{%
    breakable,
    colback=black!5,
    colframe=black!55,
    boxrule=0pt,
    leftrule=2.5pt,
    arc=0pt,
    sharp corners,
    left=8pt, right=8pt, top=5pt, bottom=5pt,
    fontupper=\small,
    before upper={\noindent\textbf{#1.}\ },
}

\begin{document}
\todototoc

\maketitle

\begin{abstract}

Decentralized Autonomous Organizations (DAOs) run protocol governance
by letting token holders vote on proposals. The dominant rule, voting power
proportional to wallet balance, concentrates control among a small number of
large holders, fueling the \emph{token-control} governance attacks
that have already compromised real protocols. To counter this concentration,
the community has turned to \emph{anti-plutocratic voting mechanisms}
such as Quadratic Voting (QV), which assign sublinear voting power per token with the goal of dampening the influence of large holders.

In this work, we {show} that no voting rule that derives power solely from wallet balance can succeed on a permissionless blockchain against a well-capitalized attacker whose objective is voting power and whose only obstacles are the per-wallet costs the chain itself imposes. 
Through a costed model of on-chain voting that captures realistic
blockchain frictions, including per-wallet splitting and voting costs, fixed setup
costs, and minimum-balance requirements, we show that whenever a wallet of any
size yields nonzero voting power, a Sybil attacker who splits tokens across
many wallets achieves total voting power that grows at least linearly in their
token holdings. For concave rules
that are actually proposed to dampen governance power---those that are positive, increasing, and finite---we show that the optimal strategy yields power that is \emph{asymptotically linear} in token holdings, \emph{regardless} of the cost scheme. 

Instantiating the model on real DAOs reveals attack costs orders of magnitude
below the value at stake. Replaying the ten most recent finalized proposals of
five major DAOs (ENS, Compound, Uniswap, Arbitrum, and ZKsync) under linear,
quadratic, logarithmic, and power-($\beta = 0.25$) voting, we measure Sybil
amplification factors between $1{,}172\times$ and $4{,}039\times$ under
Quadratic Voting, and exceeding $229{,}000\times$ under steeper power rules. On
Uniswap specifically, a Sybil-optimal attacker captures a \$300M vote for
roughly \$75K in token and gas costs; a single-wallet attacker with the same
budget achieves only $\approx 1/259$ of the voting power needed to win.


\end{abstract}

\section{Introduction}

Decentralized Autonomous Organizations (DAOs) have emerged as the foundational governance pillar of public blockchains. Today, DAOs collectively govern tens of billions of dollars across critical decentralized-finance protocols (\eg, Uniswap, ENS, Compound, Aave, MakerDAO) and Layer-2 rollup ecosystems (\eg, Arbitrum, ZKsync)~\cite{defillama}. Unlike traditional organizations managed by centralized boards, DAOs distribute protocol authority directly among holders of a native governance token. The security and legitimacy of a DAO rely on a fundamental premise: that its on-chain voting process faithfully aggregates community preferences and resists hostile manipulation. When this premise fails, protocol treasuries and smart contract parameters can be compromised irreversibly.

\pparagraph{The Fragility of Linear Voting}
Securing permissionless governance has proven notoriously difficult. Most DAOs operate under a standard \emph{linear (one-token-one-vote)} rule: each token entitles the holder to exactly one vote. While egalitarian in principle, linear voting in practice concentrates control among early investors, protocol founders, and wealthy ``whales.'' Compounding this concentration, empirical studies demonstrate widespread voter apathy: on average, only about 5\% of circulating tokens participate in any given governance vote~\cite{falk}. 

This concentration directly fuels the \emph{token control} attack vector catalogued in recent DAO security systematizations~\cite{feichtinger2024sok}. An adversary with sufficient capital can purchase tokens on secondary markets, acquire temporary liquidity via loans or flash loans, or activate dormant whale balances to force through malicious proposals. Real-world incidents demonstrate the severity of these exploits: in May 2023, an attacker cloned a proposal on \emph{Tornado Cash} with a hidden self-destruct routine, draining treasury funds~\cite{behnke2023tornado,dotan2023vulnerable}; in November 2023, an adversary accumulated NDX tokens to hijack the timelock of \emph{Indexed Finance}, forcing the DAO to burn its own treasury defensively~\cite{abrams2023indexed}. Under linear voting, capital directly buys control: whoever controls more tokens at snapshot time wins.

\pparagraph{The Promise of Concave and Quadratic Voting}
To break this plutocratic concentration, the blockchain community has turned to sublinear voting rules, most notably \emph{Quadratic Voting} (QV)~\cite{qv} and general concave power functions $f(w) = w^\beta$ ($\beta \in (0, 1)$). Under QV, committing $w$ tokens confers only $\sqrt{w}$ votes. By assigning diminishing marginal voting power to larger balances, concave rules aim to curb whale dominance: a whale holding 10,000 tokens yields only $100$ votes, allowing smaller token holders to aggregate substantial collective influence.

In traditional social choice, Lalley and Weyl~\cite{qv} established that quadratic pricing is the uniquely optimal mechanism for eliciting truthful preference intensities under price-taking conditions. Motivated by this theoretical guarantee, researchers and protocol designers have repeatedly proposed QV as an anti-plutocratic shield for on-chain DAOs~\cite{dimitri}. Crucially, however, Lalley and Weyl's optimality theorem rests on a foundational premise: that each decision-maker possesses a \emph{single, verified, and stable identity}. Indeed, the authors explicitly cautioned that high-stakes applications lacking centralized identity verification remain speculative and hazardous~\cite{qv}. 

\pparagraph{The Folklore vs. Reality on Permissionless Chains}
On permissionless blockchains, authenticated physical identities do not exist. Account creation is pseudonymously unrestricted: any user can generate arbitrarily many cryptographic keypairs at zero cost, enabling \emph{Sybil attacks}~\cite{sybil} where a participant divides their balance across multiple independent wallets.

While it is recognized as folklore that splitting an unconstrained wallet games square-root pricing~\cite{bennett,dimitri}, mechanism proponents argue that \emph{practical on-chain economic frictions} suffice to safeguard concave voting in practice: (i) gas fees to split tokens into new wallets ($s > 0$); (ii) gas fees to cast votes ($v > 0$); (iii) fixed setup costs ($p > 0$); and (iv) minimum wallet balance requirements ($m > 0$). Furthermore, designers have suggested that steeper or gentler concave curves (\eg, power-law rules $w^\beta$ or logarithmic rules $\ln(w+1)$) might alter splitting economics. Whether such frictions preserve any meaningful anti-plutocratic barrier has remained an open, unquantified question.

\pparagraph{Our Findings}
In this work, we prove that on a permissionless blockchain, \emph{no wallet-level voting rule can prevent plutocracy under a rational Sybil adversary}. Introducing a costed model capturing transaction fees, execution gas, fixed setup costs, and minimum balance thresholds, we show that whenever any balance yields nonzero voting power, an attacker's aggregate power scales at least linearly with budget ($\Omega(a)$). For standard concave rules (including quadratic, power, and logarithmic voting), the optimal splitting strategy recovers power that is \emph{asymptotically linear}, with an explicit per-token slope $\kappa$ that depends only on the mechanism and per-wallet friction.

Far from protecting the protocol, concave mechanisms create an alarming \emph{asymmetric vulnerability}: honest participants who maintain their tokens in a single wallet suffer sublinear discounting, while a well-capitalized Sybil adversary recovers linear scaling. Evaluating this exploit on 50 real-world proposals from five major DAOs reveals that attackers capture votes for three to five orders of magnitude less capital than under linear voting.

\pparagraph{Contributions}
Concretely, we make five primary contributions:
\begin{enumerate}
  \item \textbf{Cost Model of Wallet-Based Voting (\Cref{sec:model}):} We formalize on-chain voting under realistic frictions (splitting gas, voting gas, fixed setup fees, and minimum balance thresholds) for an arbitrary voting rule $f$, defining plutocratic resistance via the asymptotic scaling of Sybil-adjusted power.
  \item \textbf{Universal Linear Lower Bound (\Cref{sec:robustplutocracy}):} We prove that every nontrivial wallet-level rule is plutocratic under Sybil splitting: for any token budget $a$, an attacker achieves $V^*(a) = \Omega(a)$, independent of the functional form of $f$ or any linear cost scheme.
  \item \textbf{Tight Asymptotic Slopes for Concave Rules (\Cref{sec:proof}):} For standard concave rules, we derive the optimal splitting strategy and prove that aggregate power converges to an exact linear bound $V^*(W) \le \kappa (a - p)$, providing closed-form expressions for Quadratic, Power, and Logarithmic voting.
  \item \textbf{Empirical Evaluation on Major DAOs (\Cref{sec:impact}):} Evaluating the ten most recent finalized proposals across five major DAOs (ENS, Compound, Uniswap, Arbitrum, ZKsync), we measure Sybil amplification factors between $1{,}172\times$ and $4{,}039\times$ under QV, exceeding $229{,}000\times$ under power rules. On Uniswap, an attacker captures a \$300M vote under QV for $\approx$\$75K, whereas an honest single-wallet voter achieves only $0.38\%$ of the required power with that same capital.
  \item \textbf{Mitigations and Architectures (\Cref{sec:discussion}):} We evaluate candidate defenses (Proof-of-Personhood, economic frictions, and composite mechanisms), proving that per-wallet frictions only rescale the slope $\kappa$ without restoring sublinearity. We conclude that concave voting cannot stand alone, but must be embedded within composite architectures (\eg, bicameral chambers or identity quorums).
\end{enumerate}
\section{Preliminaries}
\label{sec:prelim}

We review the lifecycle of on-chain DAO governance and the Quadratic Voting mechanism that motivates our work.

\subsection{Decentralized Autonomous Organizations}

A \emph{Decentralized Autonomous Organization (DAO)} coordinates governance across public blockchains, stewarding treasuries exceeding \$8B across major protocols~\cite{defillama}. Decision-making authority is distributed among holders of a native \emph{governance token}.

\pparagraph{Governance Lifecycle}
Most major DAOs implement a governance architecture patterned after Compound's Governor Bravo~\cite{governorbravo}. Under this standard, each token grants one vote, establishing a default \emph{linear (one-token-one-vote)} rule. The governance process proceeds through four phases:
\begin{enumerate}
    \item \emph{Delegation}: Token holders may delegate voting power to another address (or self-delegate) without relinquishing custody.
    \item \emph{Proposal creation}: An address exceeding a predefined balance threshold submits an on-chain proposal containing executable bytecode (\eg, treasury transfers or parameter adjustments).
    \item \emph{Voting period}: Following a proposal delay, an on-chain snapshot records balances and delegations. Addresses cast votes (\emph{for}, \emph{against}, or \emph{abstain}) over several days.
    \item \emph{Timelock and execution}: If the proposal achieves a majority and satisfies a \emph{quorum} requirement, it enters a timelock delay before executing on-chain.
\end{enumerate}

\pparagraph{Plutocratic Vulnerability}
Despite the decentralized intent, linear voting concentrates governance power. Empirical studies show that DAO participation is notoriously low: typically only $\approx 5\%$ of circulating governance tokens participate in votes~\cite{falk}. Coupled with large token holdings by early investors, founders, and whales, voter apathy exposes DAOs to severe \emph{token control} attacks~\cite{feichtinger2024sok}, in which adversaries acquire, borrow, or mobilize tokens to force through malicious proposals~\cite{abrams2023indexed}. These structural vulnerabilities have led the community to seek voting mechanisms that curb the disproportionate influence of large token holders.

\subsection{Quadratic Voting and the Single-Identity Assumption}
\label{sec:prelim:qv}

To mitigate plutocracy without silencing intense voter preferences, mechanism designers proposed \emph{Quadratic Voting (QV)}~\cite{qv}. 

\pparagraph{Mechanism and Theoretical Foundation}
In standard QV, voters purchase votes using a budget of voice credits: casting $V$ votes costs $V^2$ credits. Equivalently, committing a balance of $w$ tokens yields $f(w) = \sqrt{w}$ votes. Because marginal cost scales linearly with votes cast, Lalley and Weyl~\cite{qv} prove that QV is the unique pricing rule achieving robust welfare optimality in price-taking settings: rational voters equate marginal cost to their marginal valuation, revealing preference intensities truthfully. Under square-root weighting, a whale with 100 times the balance of a small holder obtains only $10\times$ the voting power, dampening whale dominance.

\pparagraph{The Single-Identity Vulnerability}
However, Lalley and Weyl's optimality proof rests fundamentally on a \emph{single-identity assumption}: each participant controls exactly one voting account. On permissionless, pseudonymous blockchains, identity is decoupled from addresses. Anyone can generate arbitrarily many cryptographic keypairs at negligible cost. A participant holding $w$ tokens can partition their wealth across $n$ independent wallets, casting $n \sqrt{w/n} = \sqrt{n\,w}$ votes---an amplification of $\sqrt{n}$ over honest single-wallet voting.

While Dimitri~\cite{dimitri} noted this Sybil sensitivity conceptually and Bennett~\cite{bennett} analyzed Sybil resistance specifically for square-root voting, it remained unproven whether practical blockchain frictions (\eg, transaction fees, minimum balance requirements) or alternate sublinear curves can prevent Sybil exploitation. In the following sections, we formalize this problem across the entire family of concave voting rules.

\section{DAO Governance Model}
\label{sec:model}

We formalize on-chain voting under realistic blockchain frictions and introduce our metric for evaluating resistance to plutocracy under Sybil attacks.

\pparagraph{Attacker Objective and Scope}
We model an adversary with a fixed token budget $a > 0$ who seeks to maximize the total voting power they can cast by partitioning tokens across $n$ wallets. We do not assign an extrinsic utility function over the token's market price or vote outcome: governance exploits typically realize profit outside the governance system (\eg, draining a treasury or executing an oracle manipulation)~\cite{behnke2023tornado,abrams2023indexed,feichtinger2024sok}. Our objective is a rigorous worst-case analysis of governance safety: \emph{how much voting power can a given capital budget buy?} Consequently, our results provide upper bounds on adversarial capability under permissionless assumptions.

\subsection{Parameters}
\label{sec:parameters}
A DAO governance instance and a Sybil adversary are characterized by the following parameters, denominated in the DAO's governance token:
\begin{itemize}
    \item $a \in \mathbb{R}_{> 0}$: the adversary's total token budget.
    \item $W = \{w_1, \dots, w_n\} \in \mathbb{R}_{> 0}^n$: the balances of the $n$ wallets created by the attacker.
    \item $f : \mathbb{R}_{>0} \to \mathbb{R}_{>0}$: the wallet-level voting rule. Examples include linear voting ($f(w) = w$), quadratic voting ($f(w) = \sqrt{w}$), power voting ($f(w) = w^\beta$ for $\beta \in (0,1)$), and logarithmic voting ($f(w) = \ln(w+1)$).
    \item $m \in \mathbb{R}_{\ge 0}$: the minimum wallet balance required to cast a vote.
    \item $v \in \mathbb{R}_{\ge 0}$: the per-wallet voting cost (\eg, delegation and voting gas), denominated in tokens.
    \item $p \in \mathbb{R}_{\ge 0}$: the fixed, one-time setup cost of the attack (\eg, deploying a batch-splitting contract).
    \item $s \in \mathbb{R}_{\ge 0}$: the per-wallet token splitting cost (\eg, ERC-20 transfer gas).
    \item $C = (m, v, p, s)$: the cost scheme imposed on top of $f$, where at least one of $m, v, s > 0$.
\end{itemize}

\subsection{Voting Rules and Sybil-Adjusted Power}

\begin{definition}[Voting Rule]
A \emph{voting rule} is a function $f : \mathbb{R}_{> 0} \to \mathbb{R}_{> 0}$ that assigns $f(w)$ votes to a wallet holding balance $w$. We assume $f$ is nontrivial, \ie, there exists $k > 0$ such that $f(k) > 0$.
\end{definition}

An honest participant holding $a$ tokens in a single wallet receives $f(a)$ votes. A Sybil attacker, however, can distribute $a$ across $n$ wallets subject to the cost scheme.

\begin{definition}[Sybil-Adjusted Voting Power]
Given total token budget $a > 0$ and cost scheme $C = (m, v, p, s)$, the \emph{Sybil-adjusted voting power} is
\begin{equation}
\label{eq:sybil-power}
V^*(a) \;=\; \sup_{n \ge 1} \sup_{\substack{w_1, \dots, w_n \ge m \\ \sum_{i=1}^n w_i + n(s + v) + p \le a}} \sum_{i=1}^n f(w_i).
\end{equation}
\end{definition}

\begin{definition}[Plutocratic Voting Rule]
A voting rule $f$ is \emph{plutocratic} under cost scheme $C$ if there exist constants $\gamma > 0$ and $a_0 > 0$ such that for all $a \ge a_0$,
\[
V^*(a) \;\ge\; \gamma \cdot a.
\]
We refer to $\gamma$ as the asymptotic \emph{vote-yield per token}.
\end{definition}

\pparagraph{Why $\Omega(a)$ Formalizes Plutocracy}
In conventional token-weighted systems, power scales linearly with stake. Concave mechanisms were introduced specifically to make power grow \emph{sublinearly} ($f(a) = o(a)$) to curb whale influence. If a Sybil attacker can achieve $V^*(a) = \Omega(a)$, the attacker's relative advantage over an honest single-wallet holder, $V^*(a)/f(a)$, grows without bound as $a \to \infty$. Proportional voting power is precisely the plutocratic outcome these mechanisms were designed to eliminate.

\subsection{Cost Scheme Properties}
\label{sec:cost-scheme}

Under standard ERC-20 semantics, an adversary with $n$ voting wallets pays a total cost $C(n) = p + n(s + v)$.

\begin{definition}[Linear Cost Scheme]
A cost scheme $C = (m, v, p, s)$ is \emph{linear} if the total overhead for $n$ wallets satisfies $C(n) = \Omega(n)$. Any scheme with constant per-wallet fees $s + v > 0$ is linear.
\end{definition}

\pparagraph{Nonlinear Frictions}
Superlinear cost schemes ($C(n) = \omega(n)$), such as escalating per-wallet registration bonds, could in principle deter wallet splitting. However, levying a superlinear penalty requires knowing that multiple distinct addresses belong to the same entity. On a permissionless blockchain where fresh keypairs are indistinguishable from unique users, wallet-linking is fundamentally impossible without an underlying identity or Proof-of-Personhood layer (\Cref{sec:discussion}). Thus, linear cost schemes faithfully characterize permissionless DAOs.

\section{Plutocracy is Robust to Splitting Costs}
\label{sec:robustplutocracy}

To help motivate the optimal Sybil strategy for our class of concave functions, we first show that all non-trivial wallet-based voting rules are plutocratic. We
note that increasing linear rules such as $f(w)=w$ and increasing convex rules such
as $f(w)=w^2$ are obviously plutocratic because larger wallets get a natural advantage. A
more interesting case is concave rules that give voters an incentive to gain voting power by splitting their wealth into smaller wallets.

We outline the proof:

\begin{enumerate}
    \item We begin with the Sybil Voting Power Lemma  that computes a lower-bound on the voting power of a large wallet
    split into equal $m$-size chunks. \\

    \item This leads us to the Sybil-Induced Plutocracy Theorem that shows that the minimum vote-yield of this even split strategy is a constant, which means wallet-based voting is plutocratic. \\

    \item Finally, the Robustness of Plutocracy to Cost Theorem shows that all linear cost schemes reduce
    to an instance of the Sybil-Induced Plutocracy Theorem.
\end{enumerate}

\begin{restatable}[Sybil Voting Power Lemma]{lemma}{lemSybilPower}\label{lem:sybilpower}
Let $f$ be a non-trivial voting rule for which there exists $m > 0$ such that $f(m) > 0$. Then for all wallets $a \ge m$,
\[
V^*(a) \ge \left\lfloor \frac{a}{m} \right\rfloor f(m).
\]
\end{restatable}

\begin{restatable}[Sybil-Induced Plutocracy]{theorem}{thmPlutocratic}\label{theorem:plutocratic}
Every nontrivial voting rule is plutocratic under unrestricted Sybil wallet splitting.
Specifically, if $\exists m>0: f(m) > 0$, then for all $a \ge 2m$, there exists a constant $\gamma > 0$ such that $V^*(a) \ge \gamma \cdot a$.
\end{restatable}

\begin{restatable}[Robustness of Plutocracy to Costs]{theorem}{thmRobust}\label{theorem:combined}
Suppose there is a non-trivial voting scheme combined with a cost scheme $C=(m,v,p,s)$ where
there are fixed costs of additional wallet creation $p\ge 0$, a wallet splitting cost $s\ge 0$ per
wallet, a voting cost $v\ge 0$ per wallet, and a minimum wallet size $m$. The Sybil-adjusted
voting power is still plutocratic and satisfies:
\[
V_C^*(a)=\Omega(a).
\]
\end{restatable}

{Due to space constraints, all proofs are deferred to \Cref{proof:sybilpower,proof:plutocratic,proof:robust}.}

\section{Turning Concave Voting Mechanisms Linear}
\label{sec:proof}

In \Cref{sec:robustplutocracy}, we established a coarse lower bound showing that all nontrivial rules admit $\Omega(a)$ power. We now sharpen this result for the class of concave voting mechanisms---those that are positive, increasing, and concave, which encompasses all rules proposed to dampen whale dominance (\eg, quadratic, power, and logarithmic voting). We prove that an optimal Sybil attacker recovers voting power that is \emph{asymptotically linear} in their token holdings, deriving the tight asymptotic slope $\kappa$ and explicit closed-form expressions for standard mechanisms.

\subsection{Optimal Wallet Partitioning}

The Sybil attack proceeds in stages: from a budget $a$, the attacker pays fixed setup $p$, partitions the remaining $A = a - p$ across $n$ wallets, pays splitting gas $s$ per wallet, and pre-funds voting gas $v$ per wallet. Each wallet holds an effective balance $w_i \ge m$ at snapshot time and casts $f(w_i)$ votes.

We first establish that the attacker should never split tokens unequally:

\begin{restatable}[Even-Split Optimality]{lemma}{lemEven}\label{lem:even}
Let $f:\mathbb{R}_{> 0}\to\mathbb{R}_{> 0}$ be concave. Fix budget $a > 0$ and an integer $n \ge 1$. Among all allocations $(w_1,\dots,w_n)$ satisfying $w_i \ge 0$ and $\sum_{i=1}^n w_i = a$, the aggregate voting power $\sum_{i=1}^n f(w_i)$ is maximized by an even split: $w_i = a/n$ for all $i$.
\end{restatable}

The proof (deferred to \Cref{proof:even}) is a direct consequence of Jensen's inequality: concavity imposes diminishing returns per token, so any wallet with above-average wealth wastes marginal power that would yield strictly more votes if transferred to a wallet with below-average wealth.

\begin{restatable}[Optimal Token Allocation per Wallet]{lemma}{lemOptWallet}\label{lem:optimal-wallet-token-amount}
Let $f:\mathbb{R}_{> 0}\to\mathbb{R}_{> 0}$ be positive, concave, and strictly increasing. Fix total tokens $a > 0$, fixed cost $p$, and per-wallet cost $c = v + s$. For any integer $n \ge 1$ satisfying $\frac{a-p}{n} - c \ge m$, the optimal per-wallet balance is
\begin{equation}
w^* \;=\; \frac{a - p}{n} - (v + s).
\end{equation}
\end{restatable}

Since $f$ is strictly increasing, the budget constraint must bind at the optimum (\Cref{proof:optwallet}). Thus, fixing $n$ uniquely determines each wallet's balance as $w^* = \frac{a-p}{n} - c$. The attacker's optimization problem therefore reduces entirely to selecting the integer number of wallets $n$ to maximize
\begin{equation}
\label{eq:attacker-objective}
V^*(W) \;=\; \max_{n \ge 1:\, \frac{a-p}{n} - c \ge m} n \cdot f\!\left(\frac{a - p}{n} - c\right).
\end{equation}

\subsection{Asymptotic Linearity}
\label{sec:optimal-voting}

We now state the central theoretical result of the paper. Under an honest single-wallet strategy, a participant with $a$ tokens receives $f(a)$ votes, which grows strictly sublinearly for any concave $f$. Under optimal Sybil splitting, however, the attacker's power grows linearly in the budget $a$.

\begin{restatable}[Asymptotic Linearity of Concave Rules]{theorem}{thmAsympLin}\label{thm:asymp-linear}
Let $A = a - p$ and $c = v + s > 0$. Assume there exists at least one integer $n \ge 1$ such that $m \le \frac{A}{n} - c$. Define the per-token voting yield function
\begin{equation}
g(x) \;=\; \frac{f(x)}{x + c}, \qquad \text{and} \qquad \kappa \;=\; \sup_{x \ge m} g(x) < \infty.
\end{equation}
Then:
\begin{enumerate}
    \item \textbf{Linear upper bound:} For all feasible budgets $a$, the optimal voting power satisfies
    \begin{equation}
    V^*(W) \;\le\; \kappa \cdot (a - p).
    \end{equation}
    \item \textbf{Asymptotic tightness:} If $g$ attains its supremum on $[m, \infty)$ at some $x^* \ge m$, then
    \begin{equation}
    \lim_{a \to \infty} \frac{V^*(W)}{a - p} \;=\; \kappa.
    \end{equation}
\end{enumerate}
\end{restatable}

The proof (\Cref{proof:asymp}) substitutes $x = A/n - c$ into the objective: $n f(x) = \frac{A}{x+c} f(x) = A g(x) \le \kappa A$. Setting $n_A = \lfloor A / (x^* + c) \rfloor$ and applying the squeeze theorem establishes asymptotic tightness.

\pparagraph{Interpretation of the Slope $\kappa$}
The slope $\kappa = \kappa(f, c)$ is the maximum vote-yield per token. It depends solely on the choice of voting rule $f$, the minimum threshold $m$, and the per-wallet transaction friction $c = v + s$; it is independent of the budget $a$. Concavity provides a constant-factor dampening (governed by $\kappa$), but \emph{zero sublinear dampening}. A determined attacker recovers linear scaling, collapsing the mechanism back to one-token-one-vote.

\subsection{Closed Forms for Common Mechanisms}
\label{sec:closed-forms}

Relaxing $n$ to the positive reals and solving the first-order condition $f'(w^*)(w^* + c) = f(w^*)$ yields exact closed-form expressions for standard voting rules (\Cref{tab:mechanisms_Sybil}).

\begin{table}[t]
\centering
\small
\caption{Single-wallet power $V(W)$ vs. Sybil-optimal power $V^*(W)$ ($A = a - p$, $c = v + s$).}
\begin{tabular}{lll}
\hline
\textbf{Voting Function} & \textbf{Honest Baseline} & \textbf{Sybil-Optimal Power $V^*(W)$} \\
\hline
Linear Voting & $V(W) = a$ & $V^*(W) = a$ \\
Quadratic Voting & $V(W) = \sqrt{a}$ & $V^*(W) = \dfrac{a - p}{2\sqrt{v + s}}$ \\
Power Voting ($\beta \in (0, 1)$) & $V(W) = a^\beta$ & $V^*(W) = \dfrac{(a - p)\,\beta^\beta\,(1 - \beta)^{1 - \beta}}{(v + s)^{1 - \beta}}$ \\
Logarithmic Voting & $V(W) = \ln(a + 1)$ & $V^*(W) = \dfrac{(a - p)\, W_0\!\left(\frac{v + s - 1}{e}\right)}{v + s - 1}$ \\
\hline
\end{tabular}
\label{tab:mechanisms_Sybil}
\end{table}

\begin{restatable}[Closed Forms for Concave Rules]{corollary}{corClosedForms}\label{cor:closed-forms}
Let $A = a - p$ and $c = v + s$, and assume the minimum wallet bound $m$ is slack at the optimum. Under \Cref{thm:asymp-linear}, the attacker's optimal voting power satisfies:
\begin{enumerate}
    \item \textbf{Power voting ($f(x) = x^\beta$, $\beta \in (0,1)$):} 
    $V^*(W) = \frac{A\,\beta^\beta (1-\beta)^{1-\beta}}{c^{1-\beta}}$. For \emph{Quadratic Voting} ($\beta = 1/2$), this specializes to $V^*(W) = \frac{A}{2\sqrt{c}}$.
    \item \textbf{Logarithmic voting ($f(x) = \ln(x+1)$):} 
    $V^*(W) = \frac{A\,W_0((c-1)/e)}{c-1}$, where $W_0$ is the principal branch of the Lambert $W$ function.
\end{enumerate}
\end{restatable}

The derivations are deferred to \Cref{proof:closed-forms}. In all cases, $V^*(W)$ is strictly proportional to $A = a - p$, with frictions $c = v + s$ entering only into the constant denominator.

\section{Estimating the Impact on Modern DAOs}
\label{sec:impact}

We instantiate our model with on-chain governance data from five major DAOs (ENS, Compound, Uniswap, Arbitrum, and ZKsync), replaying their ten most recent finalized proposals under linear, quadratic, power ($\beta = 0.25$), and logarithmic voting. Our goal is to measure the capital reduction an attacker achieves via Sybil splitting, isolate the advantage of splitting from total budget, and evaluate whether on-chain frictions deter the exploit.

\pparagraph{Empirical Dataset and Gas Accounting}
We query address-level voting records from Tally~\cite{tally-api} for the ten most recent finalized proposals across five active DAOs as of April 2026. This yields 50 historical proposals across two distinct execution environments:
\begin{itemize}
    \item \emph{Ethereum L1 (ENS, Compound, Uniswap)}: We model transfer splitting costs at 65,000 gas (ERC-20 \texttt{transfer}) and voting costs at 175,000 gas (combined \texttt{delegate} and \texttt{castVote}). Gas prices are sampled from Etherscan~\cite{etherscan-api} on proposal creation dates.
    \item \emph{Layer-2 Rollups (Arbitrum, ZKsync)}: Because L2 transactions incur both execution gas and L1 data-posting overhead, we conservatively assign 500,000 gas for both splitting and voting, fetching historical L2 base fees via public node APIs~\cite{zksync-api}.
\end{itemize}
Token spot prices are fetched from CoinGecko~\cite{coingecko-api} on each proposal's creation date. We set $p = 0$, treating one-time attack contract deployment as negligible.

For each proposal, we aggregate the voting power of all historical participants, $\sum_i f(w_i)$, and solve for the adversarial token budget $a$ satisfying $V^*(W) = \sum_i f(w_i)$ using the closed forms in \Cref{sec:closed-forms}. Converting $a$ to USD allows direct cross-protocol comparison.

\subsection{Empirical Findings}

\begin{figure}[t]
    \centering
    \includegraphics[width=0.92\linewidth]{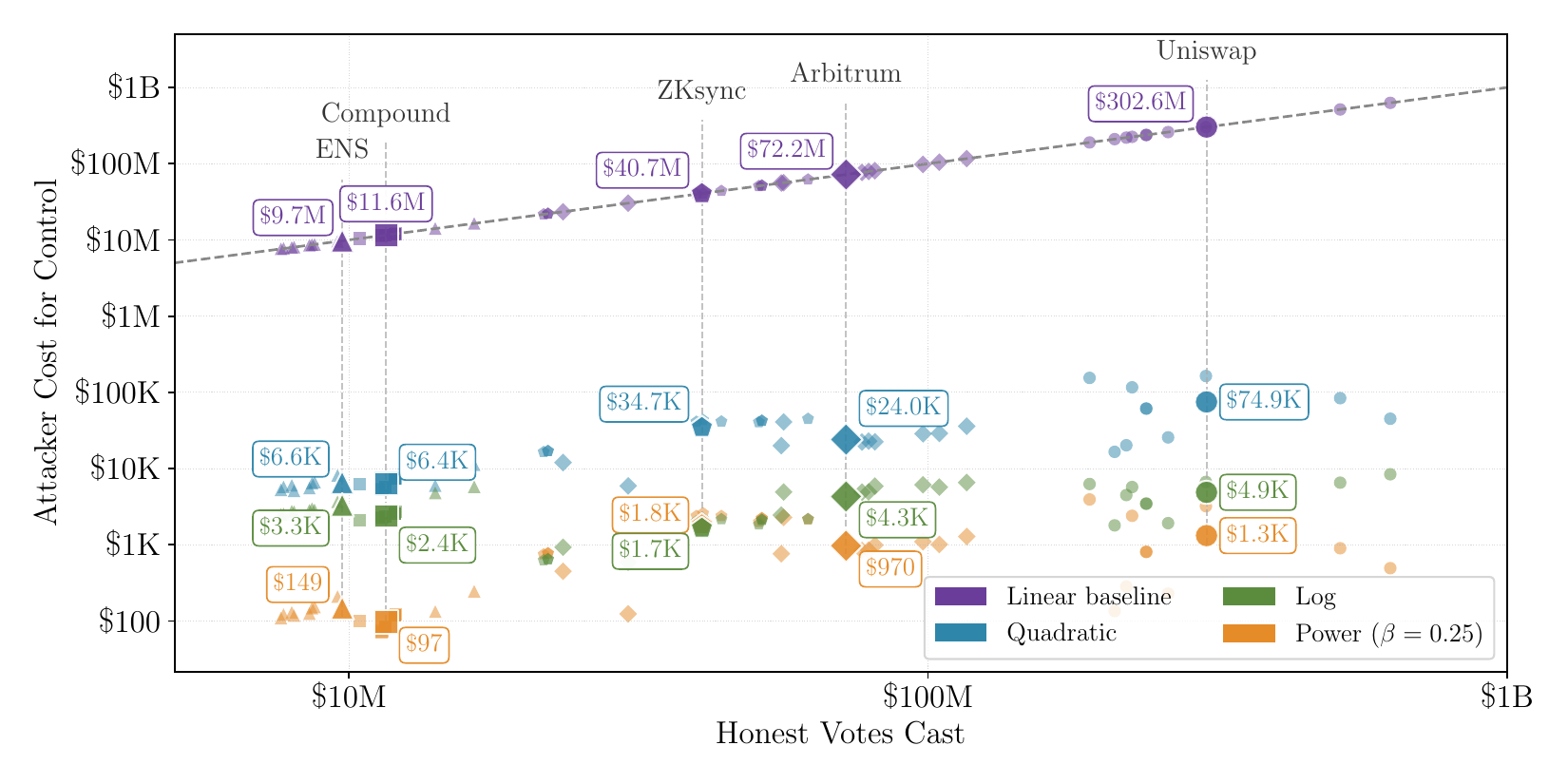}
    \caption{Adversary cost to match total honest voting power across 50 proposals from five DAOs under four voting rules. Large dots indicate protocol means; the dashed diagonal $y = x$ represents the linear baseline. Distance below the diagonal reflects the Sybil amplification factor.}
    \label{fig:sybil-scatter-all}
\end{figure}

\pparagraph{Multi-Order-of-Magnitude Capital Collapse}
\Cref{fig:sybil-scatter-all} plots the attacker's required capital against the honest linear baseline, and \Cref{tab:attacker-costs} reports the protocol averages. Under linear voting, Sybil splitting yields no advantage ($V^*(W) = a$), requiring the attacker to match honest capital dollar-for-dollar. 

Under concave voting, however, the capital required to capture the vote collapses by three to five orders of magnitude. Under Quadratic Voting, the mean attack cost drops to \$6.6K on ENS (a $1{,}472\times$ discount), \$6.4K on Compound ($1{,}820\times$), \$34.7K on ZKsync ($1{,}172\times$), \$24.0K on Arbitrum ($3{,}006\times$), and \$74.9K on Uniswap ($4{,}039\times$). Under logarithmic and power voting ($\beta = 0.25$), the collapse is even steeper: capturing Uniswap requires only \$4.9K ($62{,}041\times$) and \$1.3K ($229{,}175\times$), respectively. Rather than penalizing whales, concavity disproportionately discounts well-funded Sybil attackers.

\begin{table}[t]
\centering
\small
\caption{Mean capital cost in USD required to match all participating honest voting power across ten proposals per DAO. Parenthesized values indicate the \emph{Sybil amplification factor} relative to the linear baseline.}
\begin{tabular}{lrrrr}
\hline
\textbf{Protocol} & \multicolumn{1}{c}{\textbf{Linear}} & \multicolumn{1}{c}{\textbf{Quadratic}} & \multicolumn{1}{c}{\textbf{Logarithmic}} & \multicolumn{1}{c}{\textbf{Power ($\beta = 0.25$)}} \\
\hline
ENS      & \$9.7M   & \$6.6K (1,472$\times$)  & \$3.3K (2,912$\times$)  & \$149 (65,290$\times$) \\
Compound & \$11.6M  & \$6.4K (1,820$\times$)  & \$2.4K (4,835$\times$)  & \$97 (119,328$\times$) \\
ZKsync   & \$40.7M  & \$34.7K (1,172$\times$) & \$1.7K (24,539$\times$) & \$1.8K (22,247$\times$) \\
Arbitrum & \$72.2M  & \$24.0K (3,006$\times$) & \$4.3K (16,763$\times$) & \$970 (74,392$\times$) \\
Uniswap  & \$302.6M & \$74.9K (4,039$\times$) & \$4.9K (62,041$\times$) & \$1.3K (229,175$\times$) \\
\hline
\end{tabular}
\label{tab:attacker-costs}
\end{table}

\begin{figure}[t]
    \centering
    \includegraphics[width=0.92\linewidth]{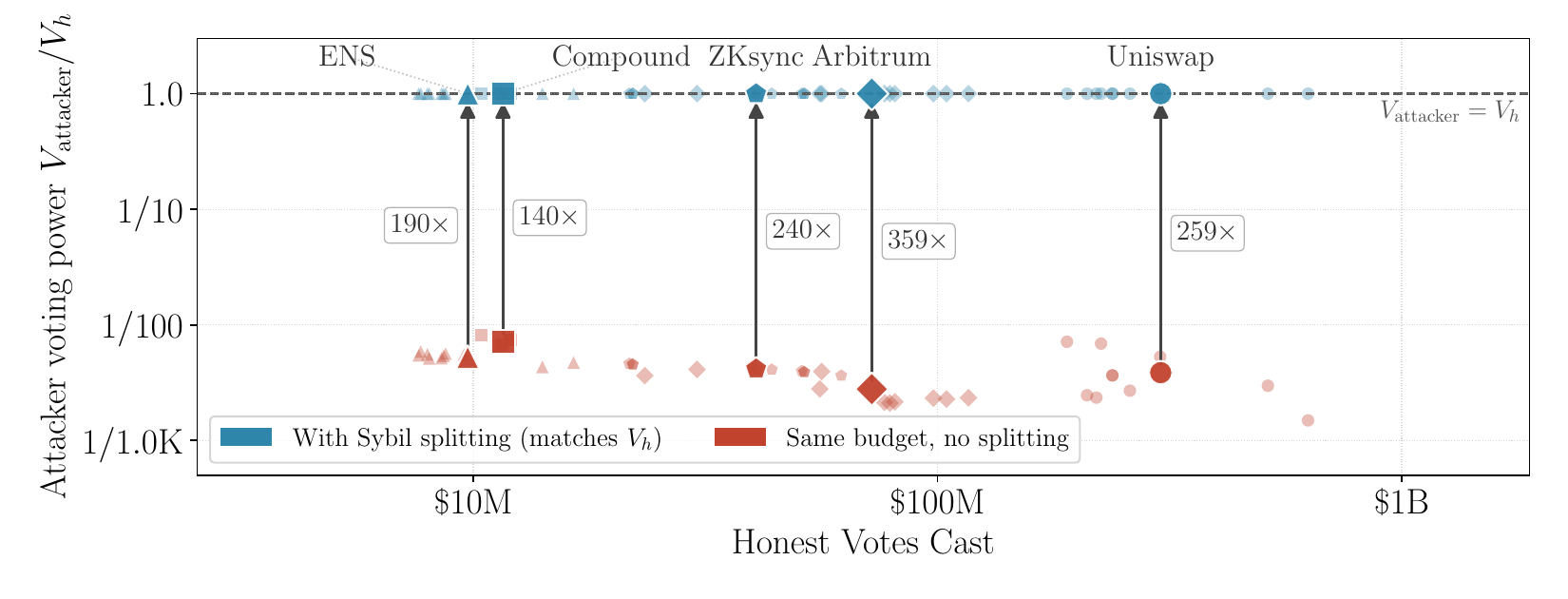}
    \caption{Attacker voting power as a fraction of honest total under Quadratic Voting at the Sybil-optimal budget from \Cref{tab:attacker-costs}. Blue dots: Sybil splitting (reaches 100\% by design). Red dots: the identical budget held in a single wallet. Upward arrows indicate the power gained purely via Sybil splitting.}
    \label{fig:sybil-voting-power}
\end{figure}

\pparagraph{Isolating the Sybil Advantage}
To isolate the power gained strictly from wallet splitting rather than the choice of curve, \Cref{fig:sybil-voting-power} fixes the adversary's USD budget at the Sybil-optimal control cost and computes the voting power that \emph{same} capital achieves when held in a single wallet under QV. On Uniswap, \$74.9K achieves 100\% of the honest voting power when partitioned across Sybil wallets, but only $\approx 1/259$ ($0.38\%$) of the required power if held in a single wallet. Across all protocols, the vertical gap confirms that the exploit is driven by splitting rather than capital volume.

\begin{figure}[h]
    \centering
    \includegraphics[width=0.92\linewidth]{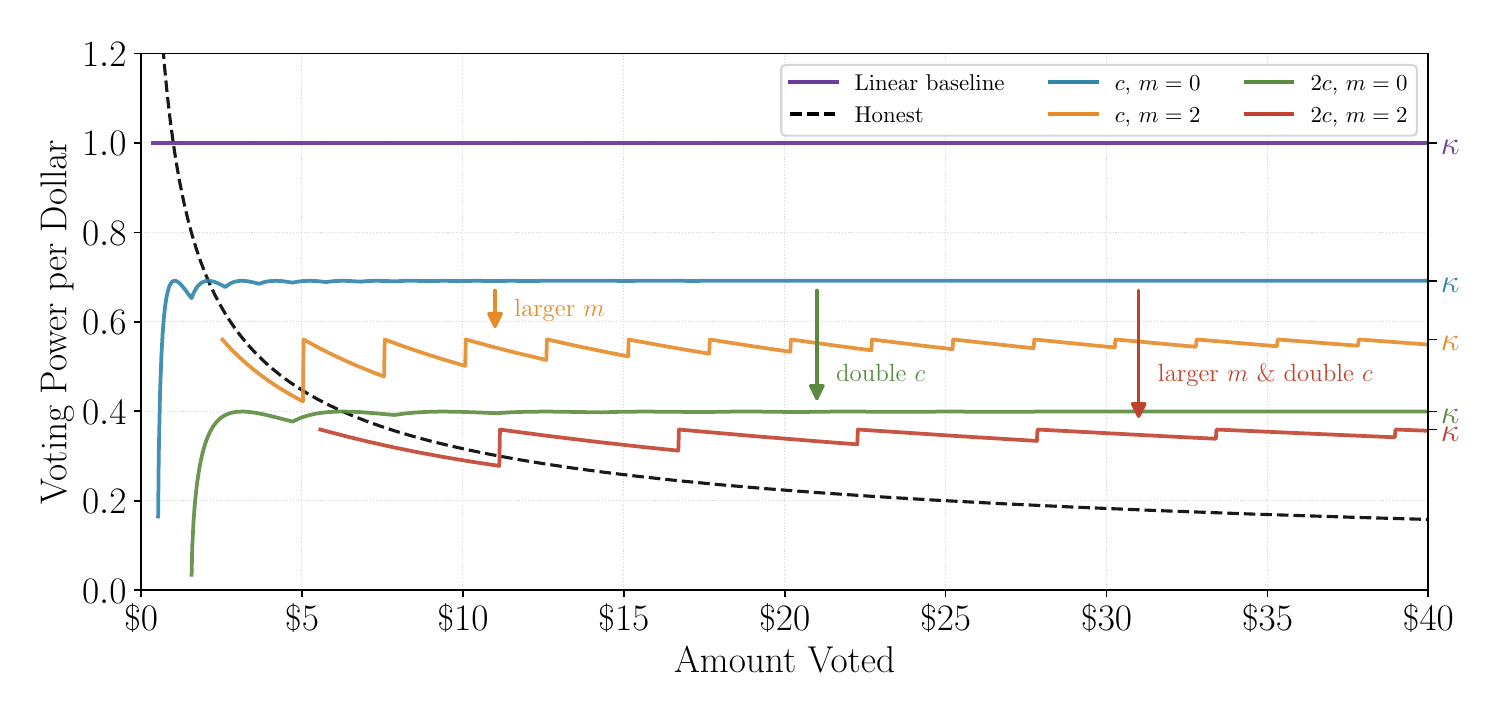}
    \caption{Voting power per dollar on Uniswap under Quadratic Voting. Dashed line: honest single wallet ($V(a)/a \to 0$). Solid lines: Sybil attacker under baseline costs ($m=0$), doubled gas ($2c$), minimum balance ($m=2$), and both. While frictions lower the plateau $\kappa$, power per dollar remains constant asymptotically.}
    \label{fig:voting-power-per-dollar}
\end{figure}

\pparagraph{Robustness to Economic Frictions}
\Cref{fig:voting-power-per-dollar} plots the voting power per dollar ($V(a)/a$ for honest voters vs. $V^*(a)/a$ for the attacker) on Uniswap under QV. For honest participants, power per dollar drops monotonically to zero as balance grows. For the Sybil attacker, power per dollar asymptotically plateaus at $\kappa = 1/(2\sqrt{c})$. Doubling gas costs ($2c$) or requiring a minimum balance ($m = 2$) reduces the plateau height, but \emph{cannot} restore sublinear scaling. Frictions simply rescale the linear constant without breaking the attack.

\pparagraph{Caveats}
Our empirical replay assumes honest participants maintain historical voting weights. In practice, a visible governance attack might stimulate defensive voter turnout, raising the capital required to capture the vote. However, given the three-to-five-order-of-magnitude discount, a Sybil attacker enjoys an overwhelming economic asymmetric advantage.

\section{Mitigation Discussion}
\label{sec:discussion}

Our results rest on a fundamental property of permissionless blockchains: an attacker can generate wallets at a per-wallet cost that is independent of their identity. Any mitigation must either break this assumption (by binding wallets to verified identities) or accept it and compose concave voting with complementary mechanism layers. No known defense preserves permissionless wallet creation while breaking asymptotic linearity against a well-capitalized adversary. We analyze candidate defenses across three primary families.

\pparagraph{1. Bounding Sybil Capacity: Proof-of-Personhood}
The most direct defense against wallet splitting binds voting rights to unique individuals. \emph{Proof-of-Personhood} (PoP) systems, such as Worldcoin (biometrics), BrightID (social graphs), and Gitcoin Passport (identity stamps), attempt to issue at most one credential per person. 

Gated by authenticated credentials, a concave voting rule successfully dampens whale influence because the number of wallets an attacker can mobilize is bounded. However, existing PoP designs encounter a fundamental trilemma across three dimensions:
\begin{enumerate}
    \item \emph{Privacy}: requiring biometric scanning or off-chain identity disclosure;
    \item \emph{Inclusivity}: disenfranchising participants unable or unwilling to enroll;
    \item \emph{Adversarial robustness}: credential farming.
\end{enumerate}
In DAO governance, stamp farming is particularly acute: when treasury assets at stake exceed millions of dollars, a well-funded attacker can synthesize or purchase credentials at secondary-market prices well below the governance payoff. PoP is therefore a \emph{necessary but not sufficient} defense: it formally breaks our linearity theorem only if the marginal cost of acquiring credentials exceeds the expected reward of governance capture.

\pparagraph{2. Raising Attack Costs: Economic Friction}
A second category of defenses introduces economic barriers to make Sybil splitting cost-prohibitive in practice, including per-wallet registration bonds, extended snapshot look-backs, and conviction or age-weighted voting (\eg, 1Hive~\cite{hive}). 

While these mechanisms increase the cost terms ($p, s$) or preparation time in our model, they do not alter the \emph{functional form} of the adversary's problem. Under \Cref{thm:asymp-linear}, $V^*(W)$ remains asymptotically linear in budget, with frictions merely reducing the slope $\kappa$. Extended look-backs and token-age weighting deter opportunistic, short-notice attacks, but an adversary planning a major exploit can pre-fund Sybil wallets months in advance. Refundable bonds increase capital lockup without increasing net capital expenditure; moreover, slashing bonds requires on-chain adjudication, which itself relies on identity verification. Finally, superlinear fee structures ($C(n) = \omega(n)$) could theoretically prevent splitting, but levying fees that escalate with an entity's wallet count requires linking distinct addresses, reducing back to the identity requirement of the first family.

\pparagraph{3. Compositional Defenses: Stacking Mechanisms}
The most pragmatic architectures accept that concave voting cannot stand alone on a permissionless blockchain and combine it with orthogonal governance layers:
\begin{itemize}
    \item \emph{Bicameral governance}: Proposals must pass two independent chambers operating under distinct rules---typically a concave chamber and a linear token-weighted chamber. An adversary who neutralizes the concave chamber via Sybil splitting remains constrained by the full capital requirement of the linear chamber.
    \item \emph{Identity quorums}: A proposal succeeds only if it clears a concave token-weight threshold and receives affirmative votes from at least $k$ distinct PoP-verified participants. This decouples preference-intensity weighting from participation quorums, preventing Sybil clusters from meeting turnout requirements.
\end{itemize}

\begin{insightbox}[Mitigation Summary]
No standalone wallet-level mechanism preserves permissionless account creation while restoring sublinear scaling against a patient adversary. Proof-of-Personhood breaks asymptotic linearity only when credential-acquisition costs exceed the value of governance capture. Economic frictions reduce the linear slope $\kappa$ but cannot restore sublinearity. Consequently, concave voting should never serve as a primary defense in isolation, but rather as one component of a composite architecture where identity verification or bicameral checks safeguard the governance perimeter.
\end{insightbox}

\section{Related Work}
\label{sec:related}

\pparagraph{Quadratic Voting and Sybil Resistance}
Lalley and Weyl~\cite{qv} established the theoretical efficiency of Quadratic Voting under price-taking assumptions, but assumed authenticated unique participants. In blockchain governance, Dimitri~\cite{dimitri} highlighted the conceptual tension between permissionless account creation and QV's single-identity premise. Bennett~\cite{bennett} formalized this for the specific case of square-root voting ($f(w) = \sqrt{w}$), proving that wallet splitting yields asymptotically linear power. Our work generalizes this result to the entire class of finite, positive, increasing concave rules, incorporates realistic on-chain cost schemes, derives tight closed-form asymptotic bounds, and quantifies the empirical vulnerability across 50 real-world DAO proposals. Kaplow and Kominers~\cite{whowillvote} analyzed turnout and coalition-based exploits under QV; in permissionless settings, an attacker does not need social coordination and can automate Sybil splitting programmatically.

\pparagraph{Cryptographic and Privacy Extensions}
Zhou \etal~\cite{10.1145/3719027.3744810} recently proposed \emph{QV-net}, introducing zero-knowledge arguments to provide maximal ballot secrecy in on-chain QV. Minaei \etal~\cite{dtl} study unlinkability for smart contracts. These cryptographic mechanisms are orthogonal to ours: because QV-net and related systems enforce governance rules at the wallet level, our impossibility theorems apply directly, leaving privacy-preserving QV equally vulnerable to Sybil splitting.

\pparagraph{DAO Governance Security}
Feichtinger \etal~\cite{feichtinger2024sok} systematized real-world DAO attacks into bribing, token control, human-computer interaction, and code vulnerabilities, while Dotan \etal~\cite{dotan2023vulnerable} studied financial incentives in DeFi governance. The vulnerability analyzed here resides in the \emph{token control} family. However, unlike traditional token-control exploits that require borrowing or purchasing outright majorities~\cite{abrams2023indexed,behnke2023tornado}, our attack exploits the voting rule itself, allowing an attacker to capture governance control with a tiny fraction of the honest capital.

\section{Conclusion}
\label{sec:conclusion}

Concave voting rules are often proposed to curb plutocracy in blockchain governance. However, our findings prove that on permissionless blockchains, concave voting mechanisms offer no lasting defense against Sybil attacks. For any positive, increasing concave rule, an attacker can split tokens across multiple wallets to recover aggregate voting power that is asymptotically linear in their token holdings. Evaluating this attack on real DAO governance data reveals that under Quadratic Voting, attackers can capture decisions with 1,100$\times$ to 4,000$\times$ less capital than under linear voting, and up to 229,000$\times$ less under steeper power rules.

Our analysis demonstrates that concavity cannot compensate for the lack of an identity layer. Protocol designers seeking anti-plutocratic governance cannot rely on wallet-level mathematical adjustments alone. Instead, as discussed in \Cref{sec:discussion}, secure governance requires composite architectures: coupling concave voting with Proof-of-Personhood credentials, identity quorums, or bicameral systems that bound Sybil influence at the protocol perimeter.

\bibliography{references}

\appendix
\crefalias{section}{appendix}
\crefalias{subsection}{subappendix}
\section{Missing Proofs}
\label{sec:missing-proofs}

We collect here the proofs of all statements in \Cref{sec:robustplutocracy,sec:proof}, in the order in which they appear in the main body. Each statement is restated before its proof.

\subsection{Proof of \texorpdfstring{\Cref{lem:sybilpower}}{the Sybil Voting Power Lemma}}
\label{proof:sybilpower}
\lemSybilPower*
\begin{proof}
Given total wealth $a$, divide it into $\left\lfloor a/m \right\rfloor$ wallets of size $m$, and one remaining wallet smaller than $m$ (which we ignore since it cannot vote). The voting power is then lower-bounded as
\[
V^*(a) \ge \left\lfloor \frac{a}{m} \right\rfloor f(m). \qedhere
\]
\end{proof}

\subsection{Proof of \texorpdfstring{\Cref{theorem:plutocratic}}{the Sybil-Induced Plutocracy Theorem}}
\label{proof:plutocratic}
\thmPlutocratic*
\begin{proof}
For $a \ge 2m$, we have
\[
\left\lfloor \frac{a}{m} \right\rfloor \ge \frac{a}{2m}.
\]
Combining with the Sybil Voting Power Lemma (\Cref{lem:sybilpower}),
\[
V^*(a) \ge \left\lfloor \frac{a}{m} \right\rfloor f(m) \ge \frac{a\cdot f(m)}{2m}.
\]
Setting $\gamma = \frac{f(m)}{2m}$ proves the result.
\end{proof}

\subsection{Proof of \texorpdfstring{\Cref{theorem:combined}}{the Robustness of Plutocracy to Costs Theorem}}
\label{proof:robust}
\thmRobust*
\begin{proof}
Assuming an attacker elects to split their budget, the fixed cost of wallet creation reduces the attacker's usable wealth from $a$ to
\[
a' = a - p.
\]
Adding in the wallet splitting cost and voting cost are equivalent to increasing the minimum wallet size:
\[
m' = m + s + v.
\]

Thus any split of the post-entry wealth $a'$ into gross wallets is an instance of \Cref{theorem:plutocratic} applied to the shifted rule $g(x)=f(x-(s+v))$ with minimum wallet size $m'$. Therefore,
\[
V_C^*(a) = \Omega(a')
\]
Since $p$ is a constant independent of $a$, we have $a - p = \Omega(a)$, and hence $V_C^*(a) = \Omega(a)$.
\end{proof}

\subsection{Proof of \texorpdfstring{\Cref{lem:even}}{the Even-Split Optimality Lemma}}
\label{proof:even}
\lemEven*
\begin{proof}
Because $f$ is concave, Jensen's inequality yields
\[
f\!\left(\frac{1}{n}\sum_{i=1}^n w_i\right) \ge \frac{1}{n}\sum_{i=1}^n f(w_i).
\]
Substituting the average wallet balance $\frac{a}{n} = \frac{1}{n}\sum_{i=1}^n w_i$ on the left and multiplying both sides by $n$ gives
\[
n\,f\!\left(\frac{a}{n}\right) \ge \sum_{i=1}^n f(w_i).
\]
Thus, for every feasible allocation $(w_1, \dots, w_n)$, the total voting power is upper-bounded by $n\,f(a/n)$, and this upper bound is achieved by the equal allocation $w_i = a/n$ for all $i$.
\end{proof}

\subsection{Proof of \texorpdfstring{\Cref{lem:optimal-wallet-token-amount}}{the Optimal Token Amount Lemma}}
\label{proof:optwallet}
\lemOptWallet*
\begin{proof}
We set up the constrained optimization problem representing the maximum voting power available to a user:
\[
V^*(W) = \max_{w_i}\ \sum_{i=1}^{n} f(w_i),
\]
subject to (i) $p + n(v + s) + \sum_{i=1}^n w_i \le a$ and (ii) $w_i \ge m,\ \forall i$. By \Cref{lem:even}, an even split of tokens among the wallets yields optimal voting power. Moreover, since $f$ is increasing, the budget constraint must bind at the optimum; otherwise, we could increase some $w_i$ while remaining feasible and strictly increase $\sum_{i=1}^n f(w_i)$. Hence, under an even split,
\begin{align*}
    p + n(v + s) + n w^* & = a \\
    w^* & = \frac{a - p}{n} - v - s.
\end{align*}
Since $\frac{a-p}{n} - v - s \ge m$ by assumption, $w_i \ge m$ for all $i$ as well as the budget constraint are satisfied.
\end{proof}

\subsection{Proof of \texorpdfstring{\Cref{thm:asymp-linear}}{the Asymptotic Linearity Theorem}}
\label{proof:asymp}
\thmAsympLin*
\begin{proof}
We prove the statement by showing each component is true.

\medskip\noindent\textbf{Step 1: Linear upper bound.} For any feasible integer $n$, let $x = \frac{A}{n} - c$ so that $x \ge m$ and $n = \frac{A}{x+c}$. By substituting $g(x)$,
\[
n\, f\!\left(\frac{A}{n} - c\right) = \frac{A}{x+c}\, f(x) = A \cdot \frac{f(x)}{x+c} = A\, g(x) \le A\,\kappa.
\]
Taking the maximum over feasible $n$ yields the upper bound:
\[
V^*(W) \le A\,\kappa.
\]

\medskip\noindent\textbf{Step 2: Asymptotic tightness.} Assume $g(x)$ attains its supremum at some $x^* \ge m$. Thus $\kappa = g(x^*)$.

For all sufficiently large $A$ such that $A \ge x^* + c$, define the following variables to conform to our integer-bounded number of wallets:
\[
n_A = \left\lfloor \frac{A}{x^* + c}\right\rfloor, \qquad x_A = \frac{A}{n_A} - c.
\]
This satisfies our lower wallet-bound condition as
\[
\frac{A}{n_A} \ge \frac{A}{A/(x^*+c)} = x^* + c,
\]
and therefore
\[
x_A = \frac{A}{n_A} - c \ge (x^* + c) - c = x^* \ge m.
\]

Next, define $t_A = \frac{A}{x^* + c}$. By definition of the floor function, $n_A \le t_A < n_A + 1$. Substituting in $t_A$, we can define upper and lower bounds, $U_A$ and $L_A$ respectively, for $x^* + c$:
\[
L_A = \frac{A}{n_A + 1} < \frac{A}{t_A} = x^* + c \le \frac{A}{n_A} = U_A.
\]
We now show that the difference between these bounds vanishes as $A \to \infty$:
\[
U_A - L_A = \frac{A}{n_A} - \frac{A}{n_A + 1} = A \cdot \frac{1}{n_A^2 + n_A}.
\]
Since $n_A = \left\lfloor \frac{A}{x^*+c}\right\rfloor \ge \frac{A}{x^*+c} - 1$, for large enough $A$ (specifically $\frac{A}{x^*+c} \ge 2$) we have $n_A \ge \frac{A}{2(x^*+c)}$. Thus,
\[
\frac{A}{n_A^2 + n_A} \le \frac{A}{n_A^2} \le \frac{A}{\bigl(A/(2(x^*+c))\bigr)^2} = \frac{4(x^*+c)^2}{A} \longrightarrow 0 \quad\text{as } A \to \infty.
\]
Since $L_A < x^* + c \le U_A$ and $U_A - L_A \to 0$, $U_A = A/n_A \to x^* + c$, and hence $x_A = A/n_A - c \to x^*$.

Because $f$ is concave it is continuous on $(0, \infty)$ (and right-continuous at $0$). Since $x + c \ge c > 0$ for all $x \ge m$, $g(x) = f(x)/(x+c)$ is continuous at $x^*$. Choose any $\varepsilon > 0$. By continuity, there exists $\delta > 0$ such that
\[
|x - x^*| < \delta \ \Rightarrow\ |g(x) - g(x^*)| < \varepsilon.
\]
Since $x_A \to x^*$, there exists $A_0$ such that for all $A \ge A_0$, $|x_A - x^*| < \delta$, and hence $|g(x_A) - g(x^*)| < \varepsilon$. This gives $\kappa - \varepsilon = g(x^*) - \varepsilon < g(x_A)$ for large $A$, and so
\[
\kappa - \varepsilon < g(x_A) = \frac{f(x_A)}{x_A + c} = \frac{n_A\, f(x_A)}{A} = \frac{n_A\, f\bigl(\tfrac{A}{n_A} - c\bigr)}{A} \le \frac{V^*(W)}{A}.
\]
Combining with $V^*(W)/A \le \kappa$ from Step 1,
\[
\kappa - \varepsilon < \frac{V^*(W)}{A} \le \kappa.
\]
Letting $\varepsilon \to 0$, the squeeze theorem gives $\lim_{a \to \infty} V^*(W)/A = \kappa$.
\end{proof}

\subsection{Proof of \texorpdfstring{\Cref{cor:closed-forms}}{the Closed-Forms Corollary}}
\label{proof:closed-forms}
\corClosedForms*
\begin{proof}
In each case, we reduce the relaxed objective $h(n) = n\, f\!\left(\tfrac{A}{n} - c\right)$ to a single-variable problem in the per-wallet amount. Let $g(w) = f(w)/(w+c)$ denote the attacker's voting power per dollar at per-wallet amount $w$. Substituting $w = A/n - c$ (so that $n = A/(w+c)$ is a bijection from $n \in (0, A/c)$ onto $w \in (0, \infty)$) gives
\[
h(n) = n\, f(w) = A \cdot \frac{f(w)}{w + c} = A \cdot g(w).
\]
The relaxed optimum is therefore $V^*(W) = A \cdot \sup_{w > 0} g(w)$. When the supremum is attained at an interior point $w^*$, that point satisfies the first-order condition $g'(w^*) = 0$. By the quotient rule,
\[
g'(w) = \frac{f'(w)\,(w + c) - f(w)}{(w + c)^2},
\]
and since $(w+c)^2 > 0$ on the interior, $g'(w^*) = 0$ is equivalent to
\[
f'(w^*)\,(w^* + c) = f(w^*). \tag{$\star$}
\]
Each case below solves $(\star)$ for $w^*$ and then evaluates $V^*(W) = A \cdot g(w^*)$.

\pparagraph{Case 1: Power voting, $f(x) = x^{\beta}$ with $\beta \in (0, 1)$}
Here $f'(w) = \beta\, w^{\beta - 1}$, so $(\star)$ becomes $\beta\, w^{\beta - 1}(w + c) = w^\beta$. Dividing through by $w^{\beta - 1}$ yields $\beta(w + c) = w$, i.e., $(1 - \beta)\,w = \beta\, c$. This gives the unique interior critical point
\[
w^* = \frac{\beta\, c}{1 - \beta}, \qquad w^* + c = \frac{c}{1 - \beta}.
\]
On the boundary, $g(w) = w^\beta/(w + c) \to 0$ as $w \to 0^+$ and as $w \to \infty$ (since $\beta < 1$), while $g(w^*) > 0$, so $w^*$ is the global maximum. Substituting,
\[
g(w^*) = \frac{(w^*)^\beta}{w^* + c} = \frac{(\beta\, c/(1-\beta))^\beta}{c/(1-\beta)} = \frac{\beta^\beta\,(1-\beta)^{1-\beta}}{c^{1-\beta}},
\]
and multiplying by $A$ yields the stated formula. The quadratic specialization is $\beta = 1/2$, which collapses $\beta^\beta(1-\beta)^{1-\beta}$ to $1/2$ and yields $V^*(W) = A/(2\sqrt{c})$.

\pparagraph{Case 2: Logarithmic voting, $f(x) = \ln(x+1)$}
Here $f'(w) = 1/(w+1)$, so $(\star)$ becomes $(w + c)/(w + 1) = \ln(w + 1)$. Substituting $x = w + 1$ (so $w + c = x + (c - 1)$) reduces the first-order condition to
\[
x\,(\ln x - 1) = c - 1.
\]
Let $y = \ln x - 1$, so that $x = e^{y+1} = e \cdot e^y$. The first-order condition becomes
\[
e \cdot y \cdot e^{y} = c - 1
\qquad\Longleftrightarrow\qquad
y\, e^{y} = \frac{c-1}{e}.
\]
By definition of the principal Lambert branch, this equation has a unique solution, which we denote
\[
y^* = W_0\!\left(\frac{c-1}{e}\right).
\]
Reversing the substitutions $y = \ln x - 1$ and $x = w + 1$ recovers the optimal per-wallet balance. From $x = e \cdot e^y$ and the identity $e^{W_0(z)} = z/W_0(z)$ (with $z = (c-1)/e$),
\[
x^* = e \cdot e^{y^*} = \frac{c - 1}{y^*},
\]
so
\[
w^* = x^* - 1 = \frac{c - 1}{y^*} - 1, \qquad w^* + c = \frac{(c-1)(1 + y^*)}{y^*}.
\]
Since $\ln(x^*) = y^* + 1$ (directly from $y^* = \ln x^* - 1$), the voting power per dollar simplifies by cancellation of the $(1 + y^*)$ factor:
\[
g(w^*) = \frac{f(w^*)}{w^* + c} = \frac{\ln(x^*)}{w^* + c} = \frac{1 + y^*}{(c-1)(1 + y^*)/y^*} = \frac{y^*}{c - 1} = \frac{W_0((c-1)/e)}{c - 1}.
\]
On the boundary, $g(w) = \ln(w+1)/(w+c) \to 0$ as $w \to 0^+$ (since $c > 0$) and as $w \to \infty$, while $g(w^*) > 0$, so $w^*$ is the global maximum. Multiplying by $A$ yields
\[
V^*(W) = \frac{A\, W_0\!\left((c-1)/e\right)}{c - 1}. \qedhere
\]
\end{proof}

\section{Honest vs Attacker Voting Power per Dollar}
\label{sec:voting-power-per-dollar}
\begin{figure}[h]
    \centering
    \includegraphics[width=\linewidth]{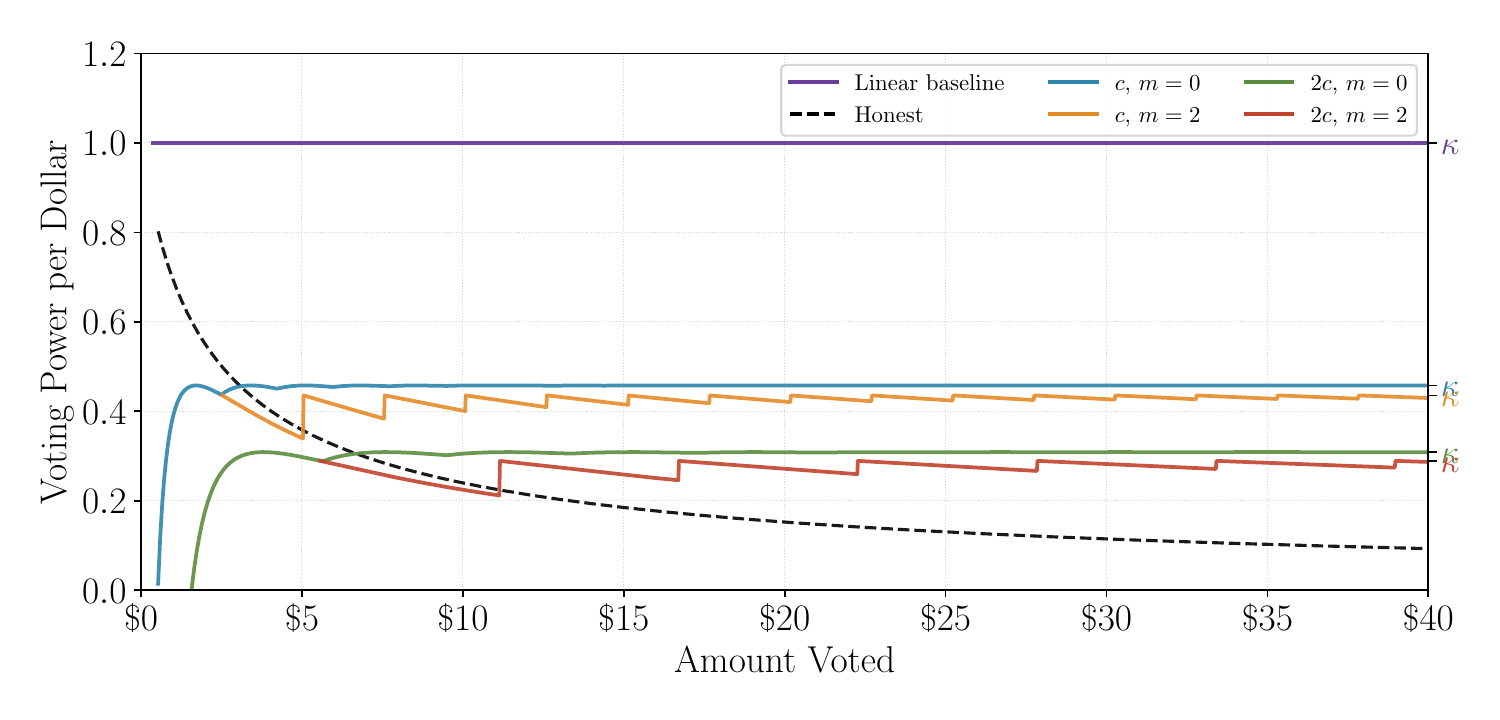}
    \includegraphics[width=\linewidth]{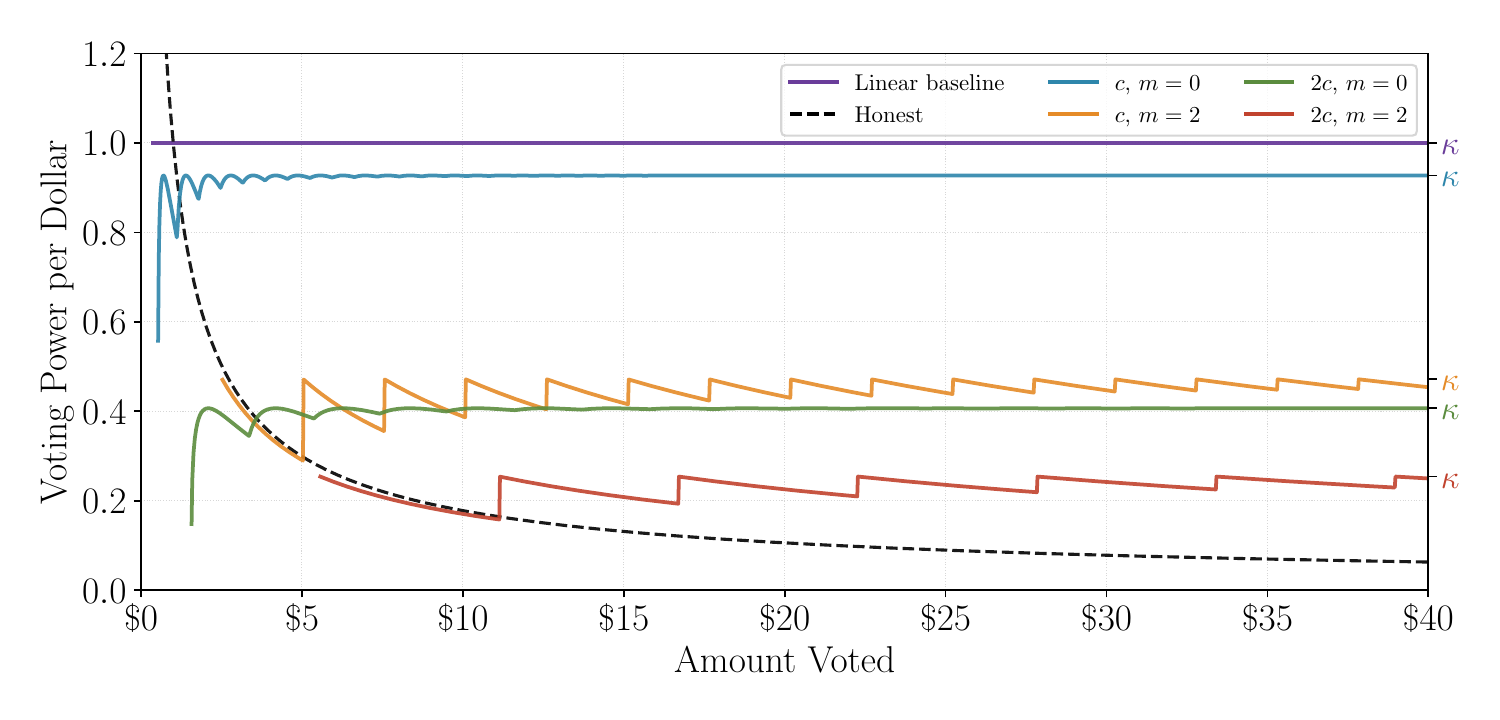}
    \caption{\textbf{Top:} Voting power per dollar for Uniswap proposals in our model under logarithmic voting. The dashed black line shows how an honest voter has less power with each additional dollar they vote with. This is true for all concave voting functions. The blue line represents an attacker under realistic conditions. The costs $c = v + s$ are taken from the analyzed proposals and current Ethereum assumptions, and there is no minimum wallet balance to vote ($m = 0$). As seen, the blue curve flattens out and approaches its $\kappa$ value. This is because the attacker achieves linear voting with enough budget. The orange, green, and red lines indicate how either doubling $c$ or increasing $m$ to 2 does not impede the attacker from achieving linear voting power, instead it just reduces the asymptotic $\kappa$ values.
    \textbf{Bottom:} The same as the top for power ($\beta = 0.25$) voting. Identical patterns are observed.}
    \label{fig:voting-power-per-dollar-log-power}
\end{figure}

\end{document}